\documentclass[letterpaper,11pt]{article}

\usepackage{cmap}
\usepackage[utf8]{inputenc}
\usepackage[T1]{fontenc}
\usepackage{amsmath,amsthm,amssymb,amsfonts}
\usepackage{mathtools} 
\usepackage{parskip}
\usepackage{enumitem}
\setlist{nosep}
\usepackage{graphicx}
\usepackage[english]{babel}
\usepackage{xcolor}
\usepackage[colorlinks,citecolor=blue,linkcolor=red]{hyperref}
\usepackage[margin=0.94in]{geometry}
\usepackage[all]{nowidow}
\usepackage[protrusion=true,expansion=true]{microtype}
\usepackage{float}
\newcommand{\eq}[1]{\hyperref[eq:#1]{(\ref*{eq:#1})}}
\renewcommand{\sec}[1]{\hyperref[sec:#1]{Section~\ref*{sec:#1}}}
\newcommand{\thm}[1]{\hyperref[thm:#1]{Theorem~\ref*{thm:#1}}}
\newcommand{\lem}[1]{\hyperref[lem:#1]{Lemma~\ref*{lem:#1}}}
\newcommand{\cor}[1]{\hyperref[cor:#1]{Corollary~\ref*{cor:#1}}}
\newcommand{\itm}[1]{\hyperref[itm:#1]{\ref*{itm:#1}}}
\newcommand{\app}[1]{\hyperref[app:#1]{Appendix~\ref*{app:#1}}}
\newcommand{\dfn}[1]{\hyperref[dfn:#1]{Definition~\ref*{dfn:#1}}}
\newcommand{\fig}[1]{\hyperref[fig:#1]{Figure~\ref*{fig:#1}}}
\newcommand{\clm}[1]{\hyperref[clm:#1]{Claim~\ref*{clm:#1}}}
\newcommand{\prot}[1]{\hyperref[prot:#1]{Protocol~\ref*{prot:#1}}}
\newcommand{\rmk}[1]{\hyperref[rmk:#1]{Remark~\ref*{rmk:#1}}}
\newcommand{\prop}[1]{\hyperref[prop:#1]{Proposition~\ref*{prop:#1}}}

\theoremstyle{definition}
\newtheorem{theorem}{Theorem}[section]
\newtheorem{lemma}[theorem]{Lemma}
\newtheorem{definition}[theorem]{Definition}
\newtheorem{proposition}[theorem]{Proposition}
\newtheorem{corollary}[theorem]{Corollary}

\newtheorem{remark}[theorem]{Remark}

\providecommand{\ket}[1]{\vert #1\rangle}
\renewcommand{\ket}[1]{\vert #1\rangle}
\providecommand{\bra}[1]{\langle #1\vert}
\renewcommand{\bra}[1]{\langle #1\vert}
\DeclareMathOperator{\tr}{tr}
\DeclareMathOperator{\supp}{supp}
\DeclareMathOperator*{\Exp}{\mathbb{E}}
\renewcommand{\epsilon}[0]{\varepsilon}
\newcommand{\Var}[0]{\mathrm{Var}}
\newcommand{\Cov}[0]{\mathrm{Cov}}

\newcommand{\class}[1]{\mathsf{#1}}
\renewcommand{\P}[0]{\class{P}}
\newcommand{\BQP}[0]{\class{BQP}}
\newcommand{\DTIME}[0]{\class{DTIME}}
\newcommand{\BQTIME}[0]{\class{BQTIME}}
\newcommand{\q}[0]{\class{q}}
\renewcommand{\r}[0]{\class{r}}
\newcommand{\bit}[0]{{\{0,1\}}}
\newcommand{\F}[0]{\mathbb{F}_2}
\newcommand{\N}[0]{\mathbb{N}}
\newcommand{\ceil}[1]{\lceil #1 \rceil}
\newcommand{\floor}[1]{\lfloor #1 \rfloor}
\newcommand{\idx}[1]{\langle #1\rangle}
\newcommand{\lossfn}[0]{\Phi}
\newcommand{\Ball}[0]{\mathcal{B}}
\newcommand{\Vol}[0]{V}
\newcommand{\nc}[0]{\mathrm{nc}}
\newcommand{\errp}[0]{\mathrm{err}}
\newcommand{\Enc}[0]{\mathsf{Enc}}
\newcommand{\Dec}[0]{\mathsf{Dec}}
\newcommand{\cdec}[0]{c_{\mathrm{dec}}}

\floatstyle{boxed}
\newfloat{boxed}{tbp}{loa}
\floatname{boxed}{Protocol}

\title{No Size-Preserving Amplification with Quantum Advice}
\author{Shih-Han Hung%
\thanks{Both authors contributed equally to this work.} \kern0.2em%
\thanks{Department of Electrical Engineering and Center for Quantum Science and Engineering, National Taiwan University. Email: shihhanh@ntu.edu.tw.}
\and
Han-Hsuan Lin
\footnotemark[1] \kern0.2em%
\thanks{National Tsing Hua University. National Center for Excellence in Quantum Information Science and Engineering. Email: linhh@cs.nthu.edu.tw.}
}
\date{}

\begin{document}

\maketitle

\begin{abstract}
Marriott and Watrous showed that quantum Merlin--Arthur games admit generic error reduction without increasing witness size [Computational Complexity, 2005]. In this work, we show that this state-size-preserving amplification property does not hold for polynomial-time quantum computation with quantum advice.
In particular, we present decision problems for which even a vanishing additive error reduction requires longer advice.
More precisely, for every polynomially bounded advice length $m(n)=\omega(1)$ and every error bound $\varepsilon(n)$ that stays below $1/2$ by at least an inverse polynomial, there is a positive function $\delta$ with $\delta(n)=O\bigl(\min\{(\log m/m)^{1/4},\ \sqrt{\log m/m}\,/(1/2-\varepsilon(n))\}\bigr)$ such that $\mathsf{BQP}_{\varepsilon}/\mathsf{q}m \subsetneq \mathsf{BQP}_{\varepsilon + \delta}/\mathsf{q}m$; for constant $\varepsilon$ the gap is $O(\sqrt{\log m/m})$.
Here, $\mathsf{BQP}_\varepsilon/\mathsf{q}m$ is the class of languages recognizable with error at most $\varepsilon(n)$ by a polynomial-time quantum algorithm with an $m(n)$-qubit advice state that only depends on the input length $n$.
We show this by proving a stronger separation $\mathsf{P}_{\varepsilon+\delta}/\mathsf{r} m \not\subset \mathsf{BQP}_{\varepsilon}/\mathsf{q} m$, where $\mathsf{P}_{\varepsilon}/\mathsf{r}m$ is the class of languages recognizable with error at most $\varepsilon(n)$ by a deterministic polynomial-time algorithm with an $m(n)$-bit advice string sampled from a distribution that depends only on $n$.
\end{abstract}

\section{Introduction}
\label{sec:introduction}

Error reduction is a basic tool in randomized and quantum computation. By repeating a bounded-error computation and taking a majority vote, one can usually obtain exponentially small error at polynomial computational cost. When a computation receives an auxiliary quantum state, however, repetition also raises a question about the amount of quantum information required: independent executions may require independent copies of a state that cannot, in general, be cloned~\cite{WZ82,Dieks82}. Marriott and Watrous~\cite{MW05} showed that quantum Merlin--Arthur verification admits a remarkable form of amplification: both completeness and soundness errors can be reduced exponentially without increasing the witness length. Their theorem makes the error parameters of $\class{QMA}$ robust even when the witness size is fixed, and yields consequences such as the equivalence of $\class{QMA}$ with logarithmic-size witnesses and $\BQP$.

This result motivates a corresponding question for computation with quantum advice. A quantum advice state supplies information that depends on the input length but is otherwise unrestricted in its preparation. Can additional computation make this information more reliable without requiring a larger advice state? The usual class $\class{BQP}/\class{qpoly}$ does not resolve this question. Its polynomial advice allowance absorbs the cost of supplying polynomially many independent copies, so standard repetition already gives error reduction within that class~\cite{NY04,Aaronson05}. To study an analogue of Marriott--Watrous amplification, one must instead keep track of the advice length itself.

We therefore write $\class{BQP}_{\epsilon}/\q m$ for the class of languages decided by a polynomial-time quantum algorithm using an advice state on at most $m(n)$ qubits, with error at most $\epsilon(n)$ on every input of length $n$. The state may be mixed and depends only on $n$. For error bounds $\epsilon<\epsilon'$, size-preserving amplification would imply
\[
    \class{BQP}_{\epsilon'}/\q m
    \subseteq
    \class{BQP}_{\epsilon}/\q m.
\]
This formulation allows both the algorithm and the advice state to change, while preserving the advice-length bound. It thus asks whether the same amount of advice can support more accurate computation, even after an arbitrary redesign of its encoding and use. A failure of this inclusion rules out any general amplification procedure satisfying that resource constraint.

There is a substantive distinction between witnesses and advice. A $\class{QMA}$ witness may depend on the particular input, and soundness must hold for every witness on a no-instance. Quantum advice must support correct decisions on all inputs of the same length, but the algorithm has no correctness obligation on other advice states~\cite{Aaronson05}. Consequently, the guarantees required for advice computation do not follow directly from the existential completeness and universal soundness guarantees amplified by Marriott and Watrous.

We show that size-preserving amplification fails for quantum advice, even for a vanishing additive improvement in error. For every polynomially bounded advice length $m(n)=\omega(1)$ and fixed target error $0<\epsilon<1/2$, our construction gives separations of the form
\[
    \class{P}_{\epsilon+\delta}/\r m
    \not\subset
    \class{BQP}_{\epsilon}/\q m,
    \qquad
    \delta(m)=O_\epsilon\!\left(\sqrt{\frac{\log m}{m}}\right),
\]
with $\epsilon+\delta(m)<1/2$ for sufficiently large $m$. Here $\mathrm{P}_{\epsilon+\delta}/r m$ denotes deterministic polynomial-time computation supplied with one $m(n)$-bit advice string sampled from a distribution that depends only on $n$; the error probability is over this sample. Such randomized classical advice is a special case of mixed-state quantum advice. Hence the separation also gives
\[
    \class{BQP}_{\epsilon}/q m
    \subsetneq
    \class{BQP}_{\epsilon+\delta}/q m.
\]
In particular, the obstruction persists even when the larger-error computation uses only classical information. Copying a sampled advice string does not supply independent samples from its advice distribution, so ordinary repetition does not remove this obstruction.

For the lower bound, we adapt the QRAC-based diagonalization of Nishimura and Yamakami~\cite[Theorem~5.1]{NY04}, which turns an information-theoretic limit on random access encoding into a lower bound on quantum advice. At a selected input length, we regard the language's answers on $k$ designated inputs as a $k$-bit string. An algorithm that answers any one of these inputs using an $m$-qubit advice state acts as a random access decoder for that string. Nayak's lower bound~\cite{Nayak99} states that encoding every $k$-bit string so that each requested bit can be recovered with error at most $\epsilon$ requires
\[
    m\geq \bigl(1-H(\epsilon)\bigr)k,
\]
where $H$ is the binary entropy function. Below this threshold, every fixed decoder must fail on at least one string, regardless of its choice of advice. Choosing such a string at a separate input length for each machine gives a language outside the smaller-error advice class. This argument concerns recovery of an individually requested bit, exactly matching the requirement imposed by advice computation.

For the larger-error computation, we need random access encodings approaching the same entropy bound. Ambainis, Nayak, Ta-Shma, and Vazirani~\cite[Theorem~5.2]{ANTV02} showed that randomized classical encodings attain this bound up to an additive logarithmic term, using covering codes and symmetrization. This identifies the information-theoretic mechanism behind the separation: a small increase in error can permit a string to fit into an advice budget that is insufficient at the smaller error. Applying this mechanism to polynomial-time computation also requires an efficient, uniform decoder and careful accounting for the information that specifies the code and its randomization.

\subsection{Technical overview}
To obtain a small separation gap, we combine a second-moment covering
argument inspired by Potukuchi and Zhang~\cite{PZ20} with Toeplitz-based
universal hashing~\cite{MNT93,Hayashi11} and the random-access-code
symmetrization of Ambainis et al.~\cite{ANTV02}.
We construct efficiently decodable random access codes from a family
of systematic linear codes whose codewords have the form $(u,Tu)$,
where $T$ is a random Toeplitz matrix.
This representation has two useful features: for codewords of length
$\ell$, the matrix can be specified using only $\ell-1$ bits, and any
requested codeword coordinate can be reconstructed efficiently from
that description and $u$.
The universality property yields a second-moment bound on the expected
distance of a uniformly random word from the random code.
Although finding a nearest codeword may be computationally expensive,
this search takes place during advice preparation, which need not be
efficient.
Decoding with the advice only reconstructs the requested
coordinate.

We then adapt the symmetrization idea of Ambainis et al.~\cite{ANTV02} and combine
it with a block construction, in which the string to be encoded is
divided into blocks of length $\ell$ and each block is encoded
separately.
A random additive mask makes each fixed input block uniformly
distributed, while a random cyclic shift spreads the expected
reconstruction error evenly over its coordinates.
Together, these operations turn the average-distance estimate into
an error guarantee for every input string and every requested bit.
Crucially, the same matrix, mask, and shift can be reused across all
blocks: the required guarantee is marginal for each query and does
not require independent errors across blocks.
Their descriptions therefore form a single shared header whose cost
is amortized over the whole encoding.

The block length controls the resulting tradeoff. Larger blocks reduce the redundancy of the underlying code, but increase the shared header and the work needed to reconstruct a coordinate. For an encoded string of length $k$ and blocks of length $\ell$, the relevant normalized overhead has order
\[
    \frac{\log\ell}{\ell}+\frac{\ell}{k}.
\]
A self-delimiting parameter header specifies $k$, $\ell$ and $d$ in $O(\log k)$ additional bits. Let $s$ be the decoding time of RAC. Noting that processing the input takes $O(n^2)$ time, the total running time is $O(n^2+s)$ while the coding overhead depends only on $s$ and $k$. Choosing $s=\Theta(\sqrt{k\log k})$ and translating the overhead through the binary entropy function yields the stated vanishing error gap for every polynomially bounded $m(n)\to\infty$.\footnote{For fixed $\varepsilon\in(0,1/2)$, the same construction also works for any sufficiently large constant advice length $m$ (depending on $\varepsilon$), with a constant gap $0<\delta<1/2-\varepsilon$. The assumption $m=\omega(1)$ ensures that our guaranteed gap tends to zero.} The construction thus connects the quantitative failure of amplification to the joint costs of compression, randomization, and efficient decoding.

\subsection{Open questions}
Our results leave open the optimal quantitative limits of size-preserving amplification.  Can separations be established for asymptotically smaller additive changes in error, and how does the answer depend on the initial error and the permitted decoding time? Conversely, is there any nontrivial amount of error reduction that is always possible at a fixed advice length, or can arbitrarily small improvements require additional advice? A converse of this kind would give a lower bound on possible separation gaps. Bounds on the redundancy of particular random access constructions would illuminate the limitations of the present approach; a general converse would require a statement about all algorithms and advice encodings.

\section{Preliminaries}

\subsection{Notation and the binary entropy function}\label{sec:entropy}

Throughout, $\log$ is the binary logarithm, $\ln$ the natural logarithm, $\N=\{1,2,\ldots\}$, $[k]=\{1,\ldots,k\}$, and $\F=\{0,1\}$ is the field with two elements; $\oplus$ denotes addition in $\F$ (bitwise XOR on strings) and $|w|$ the Hamming weight of $w\in\F^\ell$.
Vectors are column vectors, and we write $(u,v)\in\F^{d+q}$ for the concatenation of $u\in\F^d$ and $v\in\F^q$.
For $x\in\bit^n$ we write
\begin{align}\label{eq:idx}
    \idx{x} := 1+\sum_{i=1}^n x_i 2^{n-i}\in[2^n],
\end{align}
so that $x\mapsto\idx{x}$ is a bijection from $\bit^n$ onto $[2^n]$ (the rank of $x$ in lexicographic order).

The binary entropy function is $H(p)=-p\log p-(1-p)\log(1-p)$ for $p\in(0,1)$, and $H(0)=H(1)=0$.
It is continuous on $[0,1]$, strictly increasing on $[0,1/2]$ with $H(1/2)=1$, and $H^{-1}\colon[0,1]\to[0,1/2]$ denotes the inverse of its restriction to $[0,1/2]$.
We set
\begin{align}\label{eq:kappa}
    \kappa:=\sqrt{\frac{\ln 2}{2}}\approx 0.589,
    \qquad
    \bar\eta:=\frac{1}{4\kappa^2}=\frac{1}{2\ln 2}\approx 0.721,
\end{align}
and, for $\epsilon\in[0,1/2]$ and $\eta\geq 0$,
\begin{align}\label{eq:Phi}
    \lossfn(\epsilon,\eta):=\min\left\{\kappa\sqrt{\eta},\ \frac{\kappa^2\eta}{1/2-\epsilon}\right\},
\end{align}
where the second term is read as $+\infty$ when $\epsilon=1/2$.
The function $\lossfn$ measures by how much an error parameter may grow when its entropy budget grows by $\eta$ (\lem{entropy-perturbation}): it behaves like $\sqrt\eta$ in general and like $\eta/(1/2-\epsilon)$ when $\epsilon$ is bounded away from $1/2$.
We record the elementary properties that will be used.
\begin{enumerate}[label=(P\arabic*)]
    \item\label{itm:P1} $\lossfn(\epsilon,\eta)$ is nondecreasing in $\eta$, and $\lossfn(\epsilon,a\eta)\leq a\,\lossfn(\epsilon,\eta)$ for every $a\geq 1$ (because $\sqrt{a\eta}\leq a\sqrt\eta$).
    \item\label{itm:P2} $\lossfn(\epsilon,\eta)$ is nondecreasing in $\epsilon$.
    \item\label{itm:P3} If $\eta\leq\bar\eta$ then $\lossfn(\epsilon,\eta)\geq 2\kappa^2\eta$; if $\eta\geq\bar\eta$ then $\lossfn(\epsilon,\eta)\geq 1/2$.
    (Indeed $\kappa\sqrt\eta\geq 2\kappa^2\eta$ if and only if $\eta\leq\bar\eta$, and $\kappa^2\eta/(1/2-\epsilon)\geq 2\kappa^2\eta$ always; for $\eta\geq\bar\eta$ both terms are at least $\kappa\sqrt{\bar\eta}=1/2$ and $2\kappa^2\bar\eta=1/2$, respectively.)
    \item\label{itm:P4} $\lossfn(\epsilon,\eta)\leq\kappa\sqrt\eta$, and $\lossfn(\epsilon,\eta)\leq\kappa^2\eta/(1/2-\epsilon)$ when $\epsilon<1/2$.
\end{enumerate}

\begin{lemma}[Entropy perturbation]\label{lem:entropy-perturbation}
    For all $0\leq\epsilon_1\leq\epsilon_0\leq 1/2$,
    \begin{align}
        \epsilon_0-\epsilon_1\leq\lossfn\bigl(\epsilon_1,\,H(\epsilon_0)-H(\epsilon_1)\bigr).
    \end{align}
    Equivalently: if $0\leq\epsilon_1\leq\epsilon_0\leq1/2$ and $H(\epsilon_0)\leq H(\epsilon_1)+\theta$, then $\epsilon_0\leq\epsilon_1+\lossfn(\epsilon_1,\theta)$.
\end{lemma}
\begin{proof}
    Put $b_0:=1/2-\epsilon_0$, $b_1:=1/2-\epsilon_1$, $\Delta:=\epsilon_0-\epsilon_1=b_1-b_0\geq0$, and $\theta_0:=H(\epsilon_0)-H(\epsilon_1)$.
    For $x\in(0,1)$ we have $H''(x)=-\frac{1}{(\ln 2)\,x(1-x)}\leq-\frac{4}{\ln 2}$, since $x(1-x)\leq 1/4$. As $H'(1/2)=0$, for every $x\in(0,1/2]$,
    \begin{align}
        H'(x)=-\int_x^{1/2}H''(u)\,du\geq\frac{4}{\ln 2}\left(\frac{1}{2}-x\right).
    \end{align}
    Since $H$ is continuous on $[0,1/2]$ and continuously differentiable on $(0,1/2)$ with $H'\geq 0$ there,
    \begin{align}
        \theta_0=\int_{\epsilon_1}^{\epsilon_0}H'(x)\,dx
        \geq\frac{4}{\ln 2}\int_{\epsilon_1}^{\epsilon_0}\left(\frac{1}{2}-x\right)dx
        =\frac{2}{\ln 2}\bigl(b_1^2-b_0^2\bigr)
        =\frac{2}{\ln 2}\,\Delta\,(b_1+b_0)
        =\frac{\Delta\,(b_1+b_0)}{\kappa^2}.
    \end{align}
    Because $b_0\geq 0$ and $b_1\geq\Delta$, we get both $\Delta^2\leq\kappa^2\theta_0$, i.e.\ $\Delta\leq\kappa\sqrt{\theta_0}$, and (if $b_1>0$) $\Delta b_1\leq\kappa^2\theta_0$, i.e.\ $\Delta\leq\kappa^2\theta_0/(1/2-\epsilon_1)$. Hence $\Delta\leq\lossfn(\epsilon_1,\theta_0)$.
    The second formulation follows from the first by \ref{itm:P1}.
\end{proof}

The inequality $\Delta(b_1+b_0)\leq\kappa^2\theta_0$ in the proof, written as $\Delta^2+2b_0\Delta\leq\kappa^2\theta_0$, also yields $\Delta\leq\sqrt{b_0^2+\kappa^2\theta_0}-b_0\leq\kappa\theta_0/\bigl((1/2-\epsilon_0)+\sqrt{\theta_0}\bigr)$ (use $\sqrt{b_0^2+\kappa^2\theta_0}\geq\kappa\sqrt{\theta_0}$ and $\kappa<1$), which is the form used in an earlier draft; we shall only need the two bounds in \ref{itm:P4}.

The following consequence lets us move between nearby error parameters without losing more than a constant factor in the additive loss.

\begin{lemma}[Stability]\label{lem:stability}
    Let $0\leq\epsilon\leq\hat\epsilon\leq1/2$, let $\eta,\theta\geq 0$ with $H(\hat\epsilon)-H(\epsilon)\leq\theta$, and let $a\geq 1$. Then
    \begin{align}
        \min\Bigl\{\frac{1}{2},\ \hat\epsilon+a\,\lossfn(\hat\epsilon,\eta)\Bigr\}
        \leq
        \min\Bigl\{\frac{1}{2},\ \epsilon+(2a+1)\,\lossfn(\epsilon,\eta+\theta)\Bigr\}.
    \end{align}
\end{lemma}
\begin{proof}
    By \lem{entropy-perturbation}, $\hat\epsilon-\epsilon\leq\lossfn(\epsilon,\theta)$.
    We distinguish two cases.

    \emph{Case $1/2-\hat\epsilon\geq(1/2-\epsilon)/2$.}
    Comparing the two terms in \eq{Phi} gives $\lossfn(\hat\epsilon,\eta)\leq 2\lossfn(\epsilon,\eta)$ (the first term does not depend on the error parameter, and the second one at most doubles).
    Hence $\hat\epsilon+a\lossfn(\hat\epsilon,\eta)\leq\epsilon+\lossfn(\epsilon,\theta)+2a\lossfn(\epsilon,\eta)\leq\epsilon+(2a+1)\lossfn(\epsilon,\eta+\theta)$ by \ref{itm:P1}.

    \emph{Case $1/2-\hat\epsilon<(1/2-\epsilon)/2$.} Then $\hat\epsilon-\epsilon>(1/2-\epsilon)/2$, so
    $1/2<\epsilon+2(\hat\epsilon-\epsilon)\leq\epsilon+2\lossfn(\epsilon,\theta)\leq\epsilon+(2a+1)\lossfn(\epsilon,\eta+\theta)$,
    and the left-hand side of the claim is at most $1/2$.
\end{proof}

\subsection{Random access codes}\label{sec:rac}

\begin{definition}[Random access code, adapted from {\cite[Definition~4.1]{ANTV02}}]\label{dfn:rac}
    Let $k,m\geq1$ be integers and $\epsilon\in[0,1]$.
    A \emph{$k\xmapsto{\epsilon} m$ (classical) random access code} is a pair $(\Enc,\Dec)$, where $\Enc\colon\bit^k\times\Omega\to\bit^m$ for a finite set $\Omega$ of seeds equipped with a probability distribution, and $\Dec\colon[k]\times\bit^m\to\bit$ is a deterministic function, such that
    \begin{align}\label{eq:rac-condition}
        \Pr_{\omega\in\Omega}\bigl[\Dec(i,\Enc(z,\omega))=z_i\bigr]\geq1-\epsilon
        \qquad\text{for every } z\in\bit^k \text{ and every } i\in[k].
    \end{align}
    A \emph{$k\xmapsto{\epsilon}m$ quantum random access code} consists of a map $z\mapsto\rho_z$ from $\bit^k$ to density operators on $m$ qubits (mixed states are allowed) together with, for every $i\in[k]$, a two-outcome measurement $\{M_i^0,M_i^1\}$ on $m$ qubits, such that $\tr(M_i^{z_i}\rho_z)\geq1-\epsilon$ for every $z$ and $i$.
\end{definition}

Every classical code is a quantum code with the same parameters: take $\rho_z:=\Exp_\omega\ket{\Enc(z,\omega)}\!\bra{\Enc(z,\omega)}$ and $M_i^b:=\sum_{c\colon\Dec(i,c)=b}\ket{c}\!\bra{c}$.
Two trivial codes will be useful: if $m\geq k$, storing $z$ verbatim gives a $k\xmapsto{0}m$ code; and for every $k,m\geq1$ the \emph{trivial code} $\Enc(z,\omega):=\omega$ with $\omega$ uniform on $\bit^m$ and $\Dec(i,c):=c_1$ is a $k\xmapsto{1/2}m$ code, since its output is a uniformly random bit for every $z$ and $i$.

\begin{theorem}[Nayak's lower bound {\cite[Theorem~2.3, arXiv version]{Nayak99}}, see also {\cite[Theorem~5.1]{ANTV02}}]\label{thm:nayak}
    Let $\epsilon\in[0,1/2]$. Every $k\xmapsto{\epsilon}m$ quantum random access code satisfies $m\geq(1-H(\epsilon))\,k$.
    The bound holds for encodings into mixed states: the encodings in \cite[Definition~4.1]{ANTV02} may depend on internal randomness, which is the same as encoding into mixed states, and the proof of \cite[Theorem~5.1]{ANTV02} bounds the von Neumann entropy of the average encoding.
\end{theorem}

Ambainis, Nayak, Ta-Shma and Vazirani showed that the bound is tight up to an additive logarithmic term, and that this is already achieved by \emph{classical} randomized encodings: for every constant $\epsilon\in(0,1/2)$ there is a $k\xmapsto{\epsilon}m$ classical random access code with $m=(1-H(\epsilon))k+O(\log k)$ \cite[Theorem~5.2]{ANTV02}.
Their construction takes a covering code $S\subseteq\bit^k$ of covering radius about $\epsilon k$ and size $2^{(1-H(\epsilon))k+O(\log k)}$, whose existence is guaranteed by a nonconstructive argument \cite{Cohen83,CHLL97}, and symmetrizes the decoding error over the bit positions by applying a random permutation and a random mask to the message before rounding it to the nearest codeword; a further probabilistic argument fixes a small list of permutation--mask pairs so that the description of the randomness fits into $O(\log k)$ bits.
Neither the covering code nor the list is given explicitly, so the resulting decoder is not known to be efficiently computable.
In \sec{construction-rac}, we give a construction with an explicit, efficient decoder, based on random Toeplitz codes and on a shared blockwise symmetrization by a cyclic shift and a mask.

\subsection{Computational model and advice classes}\label{sec:classes}

We fix the machine model for the deterministic computations with advice; the only nonstandard ingredient is random access to the advice, which is needed as soon as the running time is smaller than the advice length.

\begin{definition}[Deterministic machines with random access to advice]\label{dfn:model}
    A \emph{machine with random access to advice} is a multitape deterministic Turing machine with a read-only input tape, work tapes, a read-only \emph{advice tape} and a work tape designated as the \emph{index tape}.
    The advice tape holds a string $a\in\bit^m$ (one bit per cell, followed by blanks). In each step the machine may move the advice head by one cell, as for any other tape, or it may execute a \emph{jump}, which in one step moves the advice head to the cell whose index is written in binary on the index tape (to the first blank cell if the index exceeds $m$).
    The running time is the number of steps. We write $M(x,a)$ for the output of $M$ on input $x$ with advice $a$.
\end{definition}

A jump followed by $d$ ordinary moves reads $d$ consecutive advice bits in $O(d+w)$ steps, where $w$ is the length of the index; arithmetic on $w$-bit integers (addition, comparison, multiplication, division with remainder) takes $O(w^2)$ steps by the schoolbook methods.
These are the only facts about the model that we use.

\begin{definition}[Advice classes]\label{dfn:classes}
    Let $t,m\colon\N\to\N$ and $\epsilon\colon\N\to[0,1]$.
    \begin{enumerate}[label=(\alph*)]
        \item A language $L\subseteq\bit^*$ is in $\DTIME(t)_\epsilon/\r m$ if there are a machine $M$ with random access to advice and a family $(D_n)_{n\in\N}$ of probability distributions, $D_n$ on $\bit^{m(n)}$, such that for every $n$ and every $x\in\bit^n$:
        (i) $\Pr_{a\sim D_n}[M(x,a)\neq L(x)]\leq\epsilon(n)$, and
        (ii) $M(x,a)$ halts within $O(t(n))$ steps for every $a\in\supp(D_n)$, with a constant that does not depend on $n$, $x$ or $a$.
        We let $\P_\epsilon/\r m:=\bigcup_{c\in\N}\DTIME(n^c)_\epsilon/\r m$.
        \item A \emph{quantum algorithm} $Q$ (for advice length $m$) is a classical Turing machine which, on input $1^n$ together with the binary representation of $m(n)$, outputs a quantum circuit $Q_n$ over a fixed finite universal gate set that contains the NOT and Toffoli gates, acting on an $n$-qubit input register, an $m(n)$-qubit advice register and ancilla qubits initialized to $\ket0$; the output $Q_n(x,\rho)$ is the result of measuring a designated output qubit in the computational basis after running $Q_n$ on $\ket x\otimes\rho\otimes\ket{0\cdots0}$. The algorithm runs in time $t$ if the classical machine runs in time $O(t(n))$ and $Q_n$ has $O(t(n))$ gates. (Gates may act on arbitrary advice qubits, so random access to the advice is built in; the generating machine is given $m(n)$ because the circuit has to be laid out on $m(n)$ advice wires, and we do not assume that $m$ is computable.)
        A language $L$ is in $\BQTIME(t)_\epsilon/\q m$ if there are a quantum algorithm running in time $t$ and density operators $\rho_n$ on $m(n)$ qubits such that $\Pr[Q_n(x,\rho_n)\neq L(x)]\leq\epsilon(n)$ for every $n$ and $x\in\bit^n$.
        We let $\BQP_\epsilon/\q m:=\bigcup_{c\in\N}\BQTIME(n^c)_\epsilon/\q m$.
    \end{enumerate}
\end{definition}

Here the advice state may be mixed, and, as in the classical case, it depends only on the input length. Note that the time bound in (a)(ii) is required for the advice strings that actually occur; if $t$ is time-constructible, the machines constructed in \sec{separation} can be clocked so that the bound holds for every advice string.

\begin{lemma}[Basic facts]\label{lem:basic}
    Let $t,m\colon\N\to\N$ with $m\geq1$ and $\epsilon,\epsilon'\colon\N\to[0,1]$.
    \begin{enumerate}[label=(\alph*)]
        \item If $\epsilon\leq\epsilon'$ pointwise, then $\DTIME(t)_\epsilon/\r m\subseteq\DTIME(t)_{\epsilon'}/\r m$, and likewise $\BQTIME(t)_\epsilon/\q m\subseteq\BQTIME(t)_{\epsilon'}/\q m$.
        \item If $\epsilon\geq1/2$ pointwise, then every language is in $\DTIME(t)_\epsilon/\r m$.
        \item If $m(n)\leq n^{O(1)}$, then $\P_\epsilon/\r m\subseteq\BQP_\epsilon/\q m$.
    \end{enumerate}
\end{lemma}
\begin{proof}
    (a) is immediate from the definitions.
    For (b), let $D_n$ be uniform on $\bit^{m(n)}$ and let $M$ output the first advice bit; then $M(x,a)$ is a uniformly random bit, so $\Pr[M(x,a)\neq L(x)]=1/2\leq\epsilon(n)$, and $M$ halts in $O(1)$ steps.
    For (c), let $L\in\DTIME(n^c)_\epsilon/\r m$ via $M$ and $(D_n)$, with $M(x,a)$ halting within $Cn^c$ steps for $a\in\supp(D_n)$, and put $\rho_n:=\sum_{a}D_n(a)\ket a\!\bra a$.
    Since $m(n)$ and $n^c$ are polynomial in $n$, the computation of $M$ on $(x,a)$ for $Cn^c$ steps (each jump being simulated by scanning the $m(n)$ advice bits) is a classical computation of polynomial size, and the generating machine, which knows $n$ and $m(n)$, can output in polynomial time a reversible circuit of NOT and Toffoli gates that performs it exactly on the input register, the advice register and ancillas, writing the output of $M$ on the output qubit (the standard reversible simulation of classical circuits by Toffoli gates); since these gates permute the computational basis, the advice register is never disturbed on basis states, and the circuit acts on $\rho_n$ as a mixture over $a\sim D_n$. Hence $\Pr[Q_n(x,\rho_n)\neq L(x)]=\Pr_{a\sim D_n}[M(x,a)\neq L(x)]\leq\epsilon(n)$.
\end{proof}

\section{Efficiently decodable random access codes}\label{sec:rac-construction}
We construct an efficiently decodable random access codes in this section. Our construction combines a second-moment covering argument inspired by Potukuchi and Zhang~\cite{PZ20} with Toeplitz-based universal hashing~\cite{MNT93, Hayashi11}. We adapt the random-access-code symmetrization of Ambainis et al. \cite{ANTV02} to a block construction in which the Toeplitz matrix, mask, and cyclic shift are shared across blocks, amortizing their description length while retaining efficient coordinate decoding.
\subsection{Random Toeplitz codes}\label{sec:toeplitz}

Fix integers $\ell\geq2$ and $d\in[\ell-1]$, and put $q:=\ell-d\geq1$.
For $t=(t_1,\ldots,t_{\ell-1})\in\F^{\ell-1}$ let $T=T(t)\in\F^{q\times d}$ be the Toeplitz matrix with entries
\begin{align}\label{eq:toeplitz}
    T_{ij}:=t_{i-j+d}\qquad(i\in[q],\ j\in[d]),
\end{align}
so that $T$ is constant along diagonals and is determined by its first row and first column; note that $i-j+d$ ranges over $[\ell-1]$.
Let $H_T:=[\,T\mid I_q\,]\in\F^{q\times\ell}$ and
\begin{align}
    C_T:=\{(u,Tu): u\in\F^d\}\subseteq\F^\ell .
\end{align}
We call $C_T$ a \emph{Toeplitz code}; it is a linear code with $2^d$ codewords, and a \emph{random Toeplitz code} is $C_{T(t)}$ for $t$ uniformly distributed on $\F^{\ell-1}$.
We write $\Pr_t$ and $\Exp_t$ for probabilities and expectations over this choice.

\begin{lemma}\label{lem:code-ht}
    For every $w\in\F^\ell$, $w\in C_T$ if and only if $H_Tw=0$.
\end{lemma}
\begin{proof}
    Write $w=(u,v)$ with $u\in\F^d$ and $v\in\F^q$. Then $H_Tw=Tu\oplus v$, which vanishes if and only if $v=Tu$, i.e., if and only if $w=(u,Tu)\in C_T$.
\end{proof}

\begin{lemma}\cite[Appendix~B]{Hayashi11}\label{lem:probability-codeword}
    For every fixed $w\in\F^\ell\setminus\{0\}$, $\Pr_t[H_{T(t)}w=0]\leq 2^{-q}$.
\end{lemma}
\begin{proof}
    Write $w=(u,v)$ as above, so that $H_Tw=0$ if and only if $Tu=v$.
    If $u=0$ then $v\neq0$ and $Tu=0\neq v$ for every $t$, so the probability is $0$.
    Assume $u\neq0$ and let $j^*:=\max\{j\in[d]:u_j=1\}$.
    With the convention $u_b:=0$ for $b\notin[d]$, for every $i\in[q]$,
    \begin{align}
        (Tu)_i=\sum_{j=1}^{d}t_{i-j+d}\,u_j=\sum_{j'=1}^{\ell-1}u_{i-j'+d}\,t_{j'},
    \end{align}
    hence $Tu=M_ut$ for the matrix $M_u\in\F^{q\times(\ell-1)}$ with entries $(M_u)_{i,j'}:=u_{i-j'+d}$.
    Consider the $q$ columns $j'_1<\cdots<j'_q$ of $M_u$ given by $j'_i:=i-j^*+d\in[\ell-1]$, and the square submatrix $N\in\F^{q\times q}$ with $N_{i,i'}:=(M_u)_{i,j'_{i'}}=u_{\,i-i'+j^*}$.
    Its diagonal entries are $N_{i,i}=u_{j^*}=1$, and for $i>i'$ we have $i-i'+j^*>j^*$ and therefore $N_{i,i'}=0$.
    Thus $N$ is upper triangular with ones on the diagonal, hence invertible, and so $M_u$ has rank $q$.
    Consequently the linear map $t\mapsto M_ut$ is surjective onto $\F^q$, all of its fibres have the same size $2^{\ell-1-q}$, and $M_ut$ is uniformly distributed on $\F^q$ when $t$ is uniform. Therefore $\Pr_t[Tu=v]=2^{-q}$.
\end{proof}

For $0\leq r\leq\ell$ let $\Ball(\ell,r):=\{w\in\F^\ell:|w|\leq r\}$ be the Hamming ball of radius $r$ and $\Vol(\ell,r):=|\Ball(\ell,r)|=\sum_{j=0}^{r}\binom{\ell}{j}$ its volume.

\begin{lemma}\label{lem:bad-syndrome}
    Let $0\leq r\leq\ell$. If $t$ is uniform on $\F^{\ell-1}$ and $s$ is uniform on $\F^q$ and independent of $t$, then
    \begin{align}
        \Pr_{t,s}\bigl[\forall w\in\Ball(\ell,r):\ H_{T(t)}w\neq s\bigr]\leq\frac{2^q}{\Vol(\ell,r)}.
    \end{align}
\end{lemma}
\begin{proof}
    Write $T=T(t)$, and for $w\in\Ball(\ell,r)$ let $Z_w$ be the indicator of the event $H_Tw=s$; let $Z:=\sum_{w\in\Ball(\ell,r)}Z_w$ be the number of elements of the ball with syndrome $s$.
    Since $s$ is uniform and independent of $t$, $\Exp[Z_w]=\Pr[H_Tw=s]=2^{-q}$ for every $w$, so $\mu:=\Exp[Z]=\Vol(\ell,r)\,2^{-q}>0$.
    For $w\neq v$ in the ball,
    \begin{align}
        \Exp[Z_wZ_v]
        &=\Pr_{t,s}[H_Tw=s\text{ and }H_Tv=s]
        =2^{-q}\sum_{s'\in\F^q}\Pr_t[H_Tw=s'\text{ and }H_Tv=s']\nonumber\\
        &=2^{-q}\Pr_t[H_Tw=H_Tv]
        =2^{-q}\Pr_t[H_T(w\oplus v)=0]
        \leq2^{-2q},
    \end{align}
    by \lem{probability-codeword} applied to $w\oplus v\neq0$. Hence $\Cov(Z_w,Z_v)=\Exp[Z_wZ_v]-\Exp[Z_w]\Exp[Z_v]\leq0$, and since each $Z_w$ is an indicator, $\Var(Z_w)\leq\Exp[Z_w]$. Therefore
    \begin{align}
        \Var(Z)=\sum_{w}\Var(Z_w)+\sum_{w\neq v}\Cov(Z_w,Z_v)\leq\sum_w\Exp[Z_w]=\mu .
    \end{align}
    The event in the statement is $\{Z=0\}\subseteq\{|Z-\mu|\geq\mu\}$, so by Chebyshev's inequality $\Pr[Z=0]\leq\Var(Z)/\mu^2\leq1/\mu=2^q/\Vol(\ell,r)$.
\end{proof}

For $y\in\F^\ell$ and a code $C\subseteq\F^\ell$ let $d(y,C):=\min_{c\in C}|y\oplus c|$ be the distance from $y$ to $C$, and define the \emph{normalized average distance}
\begin{align}\label{eq:d-ct}
    D(C):=\Exp_{y}\left[\frac{d(y,C)}{\ell}\right],\qquad y\text{ uniform on }\F^\ell .
\end{align}

\begin{corollary}\label{cor:error}
    For every $0\leq r\leq\ell$,
    \begin{align}
        \Exp_t\bigl[D(C_{T(t)})\bigr]\leq\frac{r}{\ell}+\frac{2^q}{\Vol(\ell,r)} .
    \end{align}
\end{corollary}
\begin{proof}
    Let $y$ be uniform on $\F^\ell$ and independent of $t$, and write $T=T(t)$.
    For every fixed $t$ the map $y\mapsto H_Ty$ is linear and surjective onto $\F^q$ (as $H_T$ contains the identity block), so $H_Ty$ is uniform on $\F^q$ conditioned on any value of $t$; hence $(t,H_Ty)$ has the same distribution as $(t,s)$ with $s$ uniform and independent of $t$.
    By \lem{code-ht}, $\{y\oplus c:c\in C_T\}=\{w\in\F^\ell:H_Tw=H_Ty\}$, so $d(y,C_T)=\min\{|w|:H_Tw=H_Ty\}$ and $d(y,C_T)>r$ holds if and only if no $w\in\Ball(\ell,r)$ has $H_Tw=H_Ty$.
    By \lem{bad-syndrome}, $\Pr_{t,y}[d(y,C_T)>r]\leq2^q/\Vol(\ell,r)$.
    Since $d(y,C_T)\leq r$ on the complementary event and $d(y,C_T)\leq|y\oplus0|\leq\ell$ always,
    \begin{align}
        \Exp_{t,y}[d(y,C_T)]\leq r+\ell\cdot\Pr_{t,y}[d(y,C_T)>r]\leq r+\frac{2^q\,\ell}{\Vol(\ell,r)},
    \end{align}
    and dividing by $\ell$ gives the claim.
\end{proof}

\subsection{Construction from a random Toeplitz code}\label{sec:construction-rac}

We first record the standard estimate for the volume of a Hamming ball.

\begin{proposition}\label{prop:hamming-ball}
    For all integers $0\leq r\leq\ell$, with $p:=r/\ell$,
    \begin{align}
        \Vol(\ell,r)\geq\binom{\ell}{r}\geq\frac{2^{\ell H(p)}}{\ell+1}.
    \end{align}
\end{proposition}
\begin{proof}
    If $r\in\{0,\ell\}$ then $H(p)=0$ and the claim is trivial. Otherwise $p\in(0,1)$, and by the binomial theorem $1=\sum_{j=0}^{\ell}a_j$ with $a_j:=\binom{\ell}{j}p^j(1-p)^{\ell-j}$. For $0\leq j<\ell$,
    \begin{align}
        \frac{a_{j+1}}{a_j}=\frac{(\ell-j)\,p}{(j+1)(1-p)}=\frac{(\ell-j)\,r}{(j+1)(\ell-r)},
    \end{align}
    which is at least $1$ if and only if $(\ell-j)r\geq(j+1)(\ell-r)$, i.e., if and only if $j\leq r-1+r/\ell$. As $0<r/\ell<1$, the sequence $(a_j)$ is nondecreasing for $j\leq r$ and decreasing afterwards, so $a_r$ is the largest of the $\ell+1$ summands and $a_r\geq1/(\ell+1)$.
    Finally $p^r(1-p)^{\ell-r}=2^{r\log p+(\ell-r)\log(1-p)}=2^{-\ell H(p)}$, so $\binom{\ell}{r}=a_r2^{\ell H(p)}\geq2^{\ell H(p)}/(\ell+1)$.
\end{proof}

\paragraph{Cyclic shifts.}
For $p\in[\ell]$ and $\tau\in\mathbb{Z}_\ell=\{0,\ldots,\ell-1\}$ we use the cyclic index arithmetic
\begin{align}\label{eq:shift}
    p\boxplus\tau:=\bigl((p-1+\tau)\bmod\ell\bigr)+1,\qquad
    p\boxminus\tau:=\bigl((p-1-\tau)\bmod\ell\bigr)+1,
\end{align}
so that $(p\boxminus\tau)\boxplus\tau=p$, and for every fixed $p$ the map $\tau\mapsto p\boxminus\tau$ is a bijection of $\mathbb{Z}_\ell$ onto $[\ell]$.
The cyclic shift $S_\tau\colon\F^\ell\to\F^\ell$ is defined by $(S_\tau c)_p:=c_{p\boxplus\tau}$.

\paragraph{Nearest codewords and error patterns.}
For $y\in\F^\ell$ let $\nc_T(y)\in\F^d$ be the lexicographically first $u\in\F^d$ minimizing $|(u,Tu)\oplus y|$, so that $(\nc_T(y),T\nc_T(y))$ is a nearest codeword of $C_T$ to $y$, and let
\begin{align}\label{eq:errp}
    \errp_T(y):=y\oplus\bigl(\nc_T(y),\,T\nc_T(y)\bigr)\in\F^\ell,\qquad\text{so that}\qquad|\errp_T(y)|=d(y,C_T).
\end{align}
Both $\nc_T$ and $\errp_T$ are deterministic functions of $(t,y)$.

\paragraph{Construction.}
The code is described in \prot{toeplitz-rac}.
The message $z\in\F^{B\ell}$ is cut into $B$ blocks of length $\ell$. The encoder draws one random Toeplitz matrix, one cyclic shift and one mask, which are shared by all blocks; it shifts and masks each block, rounds the result to a nearest codeword $(u_j,Tu_j)$ of the Toeplitz code and stores only the systematic part $u_j$, since $Tu_j$ is determined by $t$ and $u_j$.
Finding nearest codewords may be expensive, but this happens only during encoding, which in our application is the preparation of the advice and need not be efficient.
The decoder reconstructs the single codeword coordinate it needs from $t$ and $u_j$, undoes the mask, and never touches the other blocks.
The header $(t,\tau,\sigma)$ has $h=2\ell-1+\ceil{\log\ell}$ bits and is shared by all $B$ blocks, so its cost is amortized.

\begin{boxed}[htbp]
    \textbf{Parameters}: block length $\ell\geq2$; message length per block $d\in[\ell-1]$, with $q:=\ell-d$; number of blocks $B\geq1$. The header length is $h:=2\ell-1+\ceil{\log\ell}$ and the codeword length is $m':=h+Bd$.

    \textbf{Seed}: $(t,\tau,\sigma)$ with $t\in\F^{\ell-1}$, $\tau\in\mathbb{Z}_\ell$ and $\sigma\in\F^\ell$, independent and uniformly distributed. We write $T=T(t)$ for the Toeplitz matrix \eq{toeplitz}, $S_\tau$ for the cyclic shift \eq{shift}, and $\nc_T$ for the nearest-codeword map \eq{errp}.

    \textbf{Encoder} $\Enc(z;t,\tau,\sigma)$, on input $z\in\F^{B\ell}$:
    \begin{enumerate}[label=(E\arabic*)]
        \item Partition $z=(z_1,\ldots,z_B)$ with $z_j\in\F^\ell$.
        \item For each $j\in[B]$ let $y_j:=S_\tau(z_j)\oplus\sigma$ and $u_j:=\nc_T(y_j)\in\F^d$.
        \item Output $c:=(t,\tau,\sigma,u_1,\ldots,u_B)\in\F^{m'}$, where $\tau$ is written in binary with $\ceil{\log\ell}$ bits. Thus $t$ occupies positions $1,\ldots,\ell-1$ of $c$, $\tau$ the next $\ceil{\log\ell}$ positions, $\sigma$ the positions $\ell+\ceil{\log\ell},\ldots,h$, and $u_j$ the positions $h+(j-1)d+1,\ldots,h+jd$.
    \end{enumerate}

    \textbf{Decoder} $\Dec(a,c)$, on input an index $a\in[B\ell]$ and $c\in\F^{m'}$:
    \begin{enumerate}[label=(D\arabic*)]
        \item Let $j:=\floor{(a-1)/\ell}+1$ and $i:=((a-1)\bmod\ell)+1$, so that $z_a=(z_j)_i$.
        \item Read $\tau$ from $c$ and let $p:=i\boxminus\tau$.
        \item Read $\sigma_p$ (position $\ell+\ceil{\log\ell}-1+p$ of $c$) and the block $u_j$.
        \item If $p\leq d$, output $(u_j)_p\oplus\sigma_p$. Otherwise read $t_{p-d},\ldots,t_{p-1}$ from $c$ and output $\bigl(\sum_{y=1}^{d}t_{p-y}(u_j)_y\bigr)\oplus\sigma_p$, the sum being taken in $\F$.
    \end{enumerate}
    \caption{A random access code from a random Toeplitz code.}
    \label{prot:toeplitz-rac}
\end{boxed}

\begin{lemma}[Symmetrization]\label{lem:symmetrization}
    Fix parameters $(\ell,d,B)$ as in \prot{toeplitz-rac}. For every $z\in\F^{B\ell}$ and every $a\in[B\ell]$,
    \begin{align}\label{eq:exact-error}
        \Pr_{t,\tau,\sigma}\bigl[\Dec(a,\Enc(z;t,\tau,\sigma))\neq z_a\bigr]=\Exp_t\bigl[D(C_{T(t)})\bigr]=:\epsilon^*(\ell,d).
    \end{align}
    In particular, \prot{toeplitz-rac} is a $B\ell\xmapsto{\epsilon^*(\ell,d)}m'$ random access code, and $\epsilon^*(\ell,d)$ does not depend on $B$.
\end{lemma}
\begin{proof}
    Fix $z$ and $a$, and let $j$, $i$ be as in (D1), so that $z_a=(z_j)_i$. Fix also the seed $(t,\tau,\sigma)$, write $T=T(t)$, $y_j=S_\tau(z_j)\oplus\sigma$, $u_j=\nc_T(y_j)$, $c_j:=(u_j,Tu_j)=y_j\oplus\errp_T(y_j)$ and $p=i\boxminus\tau$.
    The decoder computes $(c_j)_p\oplus\sigma_p$: indeed, for $p\leq d$ we have $(c_j)_p=(u_j)_p$, and for $p>d$ we have $(c_j)_p=(Tu_j)_{p-d}=\sum_{y=1}^{d}T_{p-d,y}(u_j)_y=\sum_{y=1}^{d}t_{p-y}(u_j)_y$ by \eq{toeplitz}; note that $p-y\in[p-d,p-1]\subseteq[\ell-1]$.
    Since $p\boxplus\tau=i$, the definition of $S_\tau$ gives $(y_j)_p=(z_j)_{p\boxplus\tau}\oplus\sigma_p=(z_j)_i\oplus\sigma_p$, and therefore
    \begin{align}
        \Dec(a,\Enc(z;t,\tau,\sigma))=(c_j)_p\oplus\sigma_p=(y_j)_p\oplus(\errp_T(y_j))_p\oplus\sigma_p=z_a\oplus(\errp_T(y_j))_p .
    \end{align}
    Thus the decoder errs exactly when $(\errp_T(y_j))_p=1$.

    Now let the seed be random and condition on $t$ and $\tau$. The mask $\sigma$ is uniform on $\F^\ell$, so $y_j=S_\tau(z_j)\oplus\sigma$ is uniform on $\F^\ell$ as well (it is a translate of $\sigma$ by a fixed vector). Hence, for every fixed $t$, $\tau$ and every position $p'\in[\ell]$,
    \begin{align}
        \Pr_\sigma\bigl[(\errp_T(y_j))_{p'}=1\bigr]=\Pr_{y}\bigl[(\errp_T(y))_{p'}=1\bigr]=:\pi_T(p'),\qquad y\text{ uniform on }\F^\ell,
    \end{align}
    a quantity that depends neither on $z$ nor on $\tau$.
    Next let $\tau$ be uniform on $\mathbb{Z}_\ell$ and independent of $\sigma$. The position $p=i\boxminus\tau$ is then uniform on $[\ell]$, so
    \begin{align}
        \Pr_{\tau,\sigma}\bigl[\text{decoder errs}\mid t\bigr]=\frac{1}{\ell}\sum_{p'=1}^{\ell}\pi_T(p')=\frac{1}{\ell}\,\Exp_{y}\bigl[|\errp_T(y)|\bigr]=\Exp_y\left[\frac{d(y,C_T)}{\ell}\right]=D(C_T),
    \end{align}
    using \eq{errp} and \eq{d-ct}. Averaging over $t$ gives \eq{exact-error}, and the remaining claims follow from \dfn{rac}.
\end{proof}

\begin{lemma}\label{lem:encoding-length}
    Let $\ell\geq2$, $d\in[\ell-1]$, $B\geq1$, and let $q=\ell-d$, $h$ and $m'=h+Bd$ be as in \prot{toeplitz-rac}.
    \begin{enumerate}[label=(\alph*)]
        \item\label{itm:el-a} If $\epsilon\in[2/\ell,1/2]$ satisfies
        \begin{align}\label{eq:d-bound}
            \frac{d}{\ell}\geq1-H\!\left(\epsilon-\frac{2}{\ell}\right)+\frac{3\log\ell}{\ell},
        \end{align}
        then $\epsilon^*(\ell,d)\leq\epsilon$; that is, \prot{toeplitz-rac} is a $B\ell\xmapsto{\epsilon}m'$ random access code.
        \item\label{itm:el-b} There is an absolute constant $\cdec\geq1$ such that the decoder $\Dec$ can be implemented on a machine with random access to the codeword (\dfn{model}, with $c$ on the advice tape) that halts within $\cdec\,(d+w^2)$ steps, where $w\geq\max\{\log(B\ell),\log m'\}$ is an upper bound on the bit length of the index $a$, of the positions in $c$, and of the parameters $\ell$ and $d$. This holds both when $\ell$ and $d$ are hard-wired into the machine and when they are supplied to it as $w$-bit integers.
    \end{enumerate}
\end{lemma}
\begin{proof}
    \ref{itm:el-a}
    Let $r:=\floor{\epsilon\ell}-1$. Since $\epsilon\ell\geq2$ we have $r\geq1$, and
    \begin{align}\label{eq:r-range}
        \epsilon-\frac{2}{\ell}<\frac{r}{\ell}\leq\epsilon-\frac{1}{\ell}\leq\frac{1}{2}-\frac{1}{\ell} .
    \end{align}
    By \lem{symmetrization}, \cor{error} and \prop{hamming-ball},
    \begin{align}\label{eq:err-chain}
        \epsilon^*(\ell,d)=\Exp_t[D(C_{T(t)})]\leq\frac{r}{\ell}+\frac{2^q}{\Vol(\ell,r)}\leq\frac{r}{\ell}+(\ell+1)\,2^{\,q-\ell H(r/\ell)} .
    \end{align}
    By \eq{d-bound}, $q=\ell-d\leq\ell H(\epsilon-2/\ell)-3\log\ell$. By \eq{r-range}, $0\leq\epsilon-2/\ell<r/\ell\leq1/2$, and $H$ is nondecreasing on $[0,1/2]$, so $H(\epsilon-2/\ell)\leq H(r/\ell)$ and hence $q-\ell H(r/\ell)\leq-3\log\ell$. Therefore
    \begin{align}
        (\ell+1)\,2^{\,q-\ell H(r/\ell)}\leq\frac{\ell+1}{\ell^3}\leq\frac{2}{\ell^2}\leq\frac{1}{\ell}\qquad(\ell\geq2),
    \end{align}
    and \eq{err-chain} together with $r/\ell\leq\epsilon-1/\ell$ gives $\epsilon^*(\ell,d)\leq\epsilon-1/\ell+1/\ell=\epsilon$.

    \ref{itm:el-b}
    All integers handled by the decoder ($a$, $j$, $i$, $p$, $\tau$, $\ell$, $d$, $h$ and positions in $c$) have at most $w$ bits, so each arithmetic operation on them costs $O(w^2)$ steps. We go through the steps of \prot{toeplitz-rac}.
    (D1) computes a division with remainder of $a-1$ by $\ell$: $O(w^2)$ steps.
    (D2) computes the position $\ell$, jumps there and reads the $\ceil{\log\ell}\leq w$ bits of $\tau$, then computes $p=((i-1-\tau)\bmod\ell)+1$: $O(w^2)$ steps.
    (D3) computes the position of $\sigma_p$, jumps there and reads it; then computes $h+(j-1)d+1$ (one multiplication), jumps there and copies the $d$ bits of $u_j$ onto a work tape by $d$ ordinary moves: $O(w^2+d)$ steps.
    (D4) in the case $p\leq d$ locates $(u_j)_p$ on the work tape in $O(d+w)$ steps (counting down from $p$). In the case $p>d$ it writes the position $p-1$ on the index tape, jumps there, and then moves the advice head one cell to the left $d-1$ times, reading $t_{p-1},t_{p-2},\ldots,t_{p-d}$, while the head of the work tape moves over $(u_j)_1,(u_j)_2,\ldots,(u_j)_d$; the parity of $\sum_y t_{p-y}(u_j)_y$ is accumulated in the finite control. This takes $O(w+d)$ steps.
    Altogether the decoder halts within $O(d+w^2)$ steps, and the constant does not depend on the parameters. If $\ell$ and $d$ are supplied as $w$-bit integers rather than hard-wired, the only difference is that $h$ is computed from $\ell$ in $O(w^2)$ steps.
\end{proof}

\subsection{Trading decoding time for error}\label{sec:tradeoff}

We now choose the parameters of \prot{toeplitz-rac} as functions of the message length $k$, the codeword length $m$ and a bound $t$ on the number of steps of the decoder, and bound the resulting error in terms of the information-theoretic optimum $H^{-1}(1-m/k)$ of \thm{nayak}.

\begin{theorem}\label{thm:dec-error-bound}
    There are absolute constants $c_0,c_2,c\geq1$ with the following property.
    Let $k,m,t$ be positive integers with $m\leq k$ and $t\geq c_0\log^2(2k)$, and put
    \begin{align}\label{eq:eps1-eta}
        \epsilon_1:=H^{-1}\!\left(1-\frac{m}{k}\right)\in[0,1/2],\qquad
        \eta:=\frac{\log t}{t}+\frac{t}{k}.
    \end{align}
    Then there is a $k\xmapsto{\epsilon}m$ random access code with
    \begin{align}\label{eq:rac-error-upper-bound}
        \epsilon\leq\min\Bigl\{\frac{1}{2},\ \epsilon_1+c\,\lossfn(\epsilon_1,\eta)\Bigr\}
        \leq\min\Bigl\{\frac{1}{2},\ \epsilon_1+c\kappa\sqrt\eta,\ \epsilon_1+\frac{c\kappa^2\,\eta}{1/2-\epsilon_1}\Bigr\}
    \end{align}
    whose decoder halts within $t$ steps on a machine with random access to the codeword (\dfn{model}).
    Moreover, the code is either the trivial code of \sec{rac}, or \prot{toeplitz-rac} with parameters
    \begin{align}\label{eq:parameters}
        \ell=\floor{t/c_2},\qquad B=\ceil{k/\ell},\qquad d=\floor{(m-h)/B},\qquad\text{where }2\leq\ell<k,
    \end{align}
    applied to the first $k$ of its $B\ell$ message bits (the remaining ones being set to $0$) and with the codeword padded by zeros to length $m$.
\end{theorem}

\begin{proof}
    Let $\cdec$ be the constant of \lem{encoding-length}\ref{itm:el-b}, and set
    \begin{align}\label{eq:constants}
        c_2:=2\cdec,\qquad c_0:=4c_2=8\cdec,\qquad c:=10c_2+\frac{2c_2}{\kappa^2}.
    \end{align}
    Note that $c\geq 4c_2\geq1$. Put $\beta:=\epsilon_1+c\,\lossfn(\epsilon_1,\eta)$.

    \emph{Step 0: the trivial case.}
    If $\beta\geq1/2$, the trivial code has error $1/2=\min\{1/2,\beta\}$, and its decoder (write the index $1$ on the index tape, jump, read, output) halts within an absolute constant number of steps, which we may assume to be at most $c_0$ (increase $\cdec$ if necessary); hence it halts within $t\geq c_0$ steps. So assume from now on that $\beta<1/2$. Then $c\lossfn(\epsilon_1,\eta)<1/2$, and \ref{itm:P3} (with $c\geq1$) forces
    \begin{align}\label{eq:eta-small}
        \eta<\bar\eta<1;\qquad\text{in particular}\qquad t<k.
    \end{align}

    \emph{Step 1: parameters.}
    Let $\ell:=\floor{t/c_2}$. As $t\geq c_0=4c_2$, we have $\ell\geq4$, $\ell\leq t/c_2\leq t/2$ and $\ell\geq t/c_2-1\geq t/(2c_2)$; by \eq{eta-small}, $\ell<k$.
    Let $B:=\ceil{k/\ell}$, so that $k\leq B\ell<k+\ell<2k$ and $B\leq k/\ell+1\leq2k/\ell$; let $h:=2\ell-1+\ceil{\log\ell}\leq3\ell$ and $d:=\floor{(m-h)/B}$.
    From $t/(2c_2)\leq\ell\leq t$ and $\log t\geq1$ we obtain
    \begin{align}\label{eq:ell-vs-t}
        \frac{\ell}{k}\leq\frac{t}{k},\qquad
        \frac{\log\ell}{\ell}\leq\frac{2c_2\log t}{t},\qquad
        \frac{1}{\ell}\leq\frac{2c_2}{t}\leq\frac{2c_2\log t}{t}.
    \end{align}

    \emph{Step 2: $1\leq d\leq\ell-1$.}
    Since $m\leq k\leq B\ell$ and $h\geq1$, we have $(m-h)/B<\ell$, hence $d\leq\ell-1$.
    Suppose that $d\leq0$, i.e., $m<h+B$. Then, by \eq{ell-vs-t},
    \begin{align}
        1-H(\epsilon_1)=\frac{m}{k}<\frac{3\ell+2k/\ell}{k}=\frac{3\ell}{k}+\frac{2}{\ell}\leq\frac{3t}{k}+\frac{4c_2\log t}{t}\leq4c_2\,\eta .
    \end{align}
    \lem{entropy-perturbation} with $\epsilon_0=1/2$ and \ref{itm:P1} give $1/2-\epsilon_1\leq\lossfn(\epsilon_1,1-H(\epsilon_1))\leq\lossfn(\epsilon_1,4c_2\eta)\leq4c_2\lossfn(\epsilon_1,\eta)\leq c\lossfn(\epsilon_1,\eta)$, i.e., $\beta\geq1/2$, contradicting our assumption. Hence $d\geq1$, and in particular $m\geq h+B$.

    \emph{Step 3: the code and its decoding time.}
    By \lem{symmetrization}, \prot{toeplitz-rac} with parameters $(\ell,d,B)$ is a $B\ell\xmapsto{\epsilon^*}m'$ random access code with $\epsilon^*:=\epsilon^*(\ell,d)$ and $m'=h+Bd\leq m$.
    Restricting the message to its first $k$ bits (the encoder sets the other $B\ell-k$ bits to $0$) and padding every codeword with $m-m'$ zeros yields a $k\xmapsto{\epsilon^*}m$ random access code; the decoder ignores the padding.
    Let $w:=\ceil{\log(2k)}$. Then $w\geq\log(B\ell)$ because $B\ell<2k$, $w\geq\log m'$ because $m'\leq m\leq k$, and $w$ bounds the bit length of $\ell,d<k$. Since $\log(2k)\geq1$, $w^2\leq(\log(2k)+1)^2\leq4\log^2(2k)\leq4t/c_0$. By \lem{encoding-length}\ref{itm:el-b} the decoder halts within
    \begin{align}
        \cdec\,(d+w^2)\leq\cdec\Bigl(\frac{t}{c_2}+\frac{4t}{c_0}\Bigr)=\frac{t}{2}+\frac{t}{2}=t
    \end{align}
    steps.

    \emph{Step 4: the error bound.}
    Let
    \begin{align}
        y:=1-\frac{d}{\ell}+\frac{3\log\ell}{\ell}>0,\qquad
        \epsilon_0:=H^{-1}\bigl(\min\{1,y\}\bigr)\in[0,1/2],\qquad
        \epsilon:=\epsilon_0+\frac{2}{\ell} .
    \end{align}
    We first bound $\epsilon_0$ in terms of $\epsilon_1$. From $d=\floor{(m-h)/B}$ we get $Bd\geq m-h-B$, and $m-h-B\geq0$ by Step~2. Using $B\ell<k+\ell$, $\frac{1}{k+\ell}=\frac{1}{k}\bigl(1-\frac{\ell}{k+\ell}\bigr)\geq\frac{1}{k}\bigl(1-\frac{\ell}{k}\bigr)$ and $(m-h-B)/k\leq1$,
    \begin{align}
        \frac{d}{\ell}\geq\frac{m-h-B}{B\ell}\geq\frac{m-h-B}{k+\ell}\geq\frac{m-h-B}{k}\Bigl(1-\frac{\ell}{k}\Bigr)\geq\frac{m}{k}-\frac{h+B+\ell}{k}\geq\frac{m}{k}-\frac{4\ell}{k}-\frac{2}{\ell},
    \end{align}
    where the last step uses $h\leq3\ell$ and $B\leq2k/\ell$. Consequently, by \eq{ell-vs-t},
    \begin{align}
        y\leq1-\frac{m}{k}+\frac{4\ell}{k}+\frac{2+3\log\ell}{\ell}
        \leq H(\epsilon_1)+\frac{4t}{k}+\frac{10c_2\log t}{t}
        \leq H(\epsilon_1)+10c_2\,\eta .
    \end{align}
    Hence $H(\epsilon_0)\leq\min\{1,y\}\leq H(\epsilon_1)+10c_2\eta$. If $\epsilon_0\leq\epsilon_1$ there is nothing to prove; otherwise \lem{entropy-perturbation} applies. In both cases, using \ref{itm:P1},
    \begin{align}\label{eq:eps0-bound}
        \epsilon_0\leq\epsilon_1+\lossfn(\epsilon_1,10c_2\eta)\leq\epsilon_1+10c_2\,\lossfn(\epsilon_1,\eta).
    \end{align}
    Next, by \eq{ell-vs-t}, \eq{eta-small} and \ref{itm:P3},
    \begin{align}\label{eq:two-over-ell}
        \frac{2}{\ell}\leq\frac{4c_2\log t}{t}\leq4c_2\,\eta\leq\frac{2c_2}{\kappa^2}\,\lossfn(\epsilon_1,\eta).
    \end{align}
    Adding \eq{eps0-bound} and \eq{two-over-ell} gives
    \begin{align}\label{eq:eps-final}
        \epsilon=\epsilon_0+\frac{2}{\ell}\leq\epsilon_1+\Bigl(10c_2+\frac{2c_2}{\kappa^2}\Bigr)\lossfn(\epsilon_1,\eta)=\beta<\frac{1}{2} .
    \end{align}
    In particular $\epsilon_0<1/2$, so $H(\epsilon_0)<1$, which forces $\min\{1,y\}=y$ and $H(\epsilon_0)=y$; that is,
    \begin{align}
        \frac{d}{\ell}=1-H\Bigl(\epsilon-\frac{2}{\ell}\Bigr)+\frac{3\log\ell}{\ell}.
    \end{align}
    As $2/\ell\leq\epsilon\leq1/2$, \lem{encoding-length}\ref{itm:el-a} yields $\epsilon^*\leq\epsilon\leq\beta$.
    Together with Step~3 this proves the first inequality in \eq{rac-error-upper-bound}; the second one is \ref{itm:P4}.
\end{proof}

\paragraph{Discussion.}
In \eq{rac-error-upper-bound} the first term, $H^{-1}(1-m/k)$, is a lower bound on the error of \emph{every} $k\xmapsto{\epsilon}m$ random access code, classical or quantum, by \thm{nayak}.
The second term is the additive loss caused by restricting the decoder to $t$ steps and by the shared header.
It is governed by $\eta=(\log t)/t+t/k$: the first summand decreases with the block length (which grows with $t$), the second accounts for the header, which grows with the block length while the number of blocks shrinks.
The two summands are balanced at $t_0:=\ceil{\sqrt{k\log k}}$: for $k\geq4$ one checks $(\log t_0)/t_0\leq2\sqrt{(\log k)/k}$ and $t_0/k\leq2\sqrt{(\log k)/k}$, so $\eta\leq4\sqrt{(\log k)/k}$ and the loss is $O((\log k/k)^{1/4})$ in general, and $O\bigl(\sqrt{(\log k)/k}\,/(1/2-\epsilon_1)\bigr)$ when $\epsilon_1$ is bounded away from $1/2$. Note that $t_0\geq c_0\log^2(2k)$ for all $k$ larger than a constant.
When the running time is more restricted, $t=o(\sqrt{k\log k})$, the term $(\log t)/t$ dominates and the loss is $O(\sqrt{(\log t)/t})$, respectively $O\bigl((\log t)/(t(1/2-\epsilon_1))\bigr)$.

\section{No size-preserving amplification}\label{sec:separation}

In this section we prove separations between advice classes with the same advice length and different error bounds.
Fix a function $k\colon\N\to\N$ with $k(n)\leq2^n$ for all $n$. For $R\colon\bit^*\to\bit$ define the language
\begin{align}\label{eq:separation-language-2}
    L^{(k)}_R:=\bigl\{x\in\bit^*:\ \idx{x}\leq k(|x|)\text{ and }R(x)=1\bigr\},
\end{align}
where $\idx{x}$ is the lexicographic rank \eq{idx}: on inputs of length $n$ the language agrees with $R$ on the first $k(n)$ strings of $\bit^n$ and is $0$ on the remaining ones.
For $n\in\N$ let $z_n(R)\in\bit^{k(n)}$ be the truth table of $R$ on these strings, $z_n(R)_j:=R(x)$ for the $x\in\bit^n$ with $\idx{x}=j$. Thus
\begin{align}\label{eq:L-vs-z}
    L^{(k)}_R(x)=z_n(R)_{\idx{x}}\qquad\text{whenever }|x|=n\text{ and }\idx{x}\leq k(n).
\end{align}
This family of languages allows a diagonalization argument: at a given input length, deciding $L^{(k)}_R$ from an advice state amounts to decoding single bits of $z_n(R)$ from that state.

\subsection{The lower bound}

\begin{theorem}\label{thm:soundness-2}
    Let $m,k\colon\N\to\N$ with $k(n)\leq2^n$, and let $\epsilon\colon\N\to[0,1/2]$. If
    \begin{align}\label{eq:nayak-condition}
        m(n)<\bigl(1-H(\epsilon(n))\bigr)\,k(n)\qquad\text{for infinitely many }n,
    \end{align}
    then there is an $R\colon\bit^*\to\bit$ such that $L^{(k)}_R\notin\BQTIME(t)_\epsilon/\q m$ for every $t\colon\N\to\N$. In particular $L^{(k)}_R\notin\BQP_\epsilon/\q m$.
\end{theorem}
\begin{proof}
    Let $S$ be the (infinite) set of $n$ satisfying \eq{nayak-condition}, and let $n_1<n_2<\cdots$ be its elements.
    The quantum algorithms of \dfn{classes}(b) for advice length $m$ are given by classical Turing machines (those that output a circuit on every input $1^n$ with $m(n)$), so there are countably many of them; let $Q^{(1)},Q^{(2)},\ldots$ be an enumeration of all of them, with no restriction on their running time. (The enumeration need not be effective, as $R$ is not required to be computable.)

    Fix $i$ and put $n:=n_i$. For $z\in\bit^{k(n)}$, say that a density operator $\rho$ on $m(n)$ qubits is \emph{good for $z$} (with respect to $Q^{(i)}$) if
    \begin{align}
        \Pr\bigl[Q^{(i)}_n(x,\rho)\neq z_{\idx{x}}\bigr]\leq\epsilon(n)\qquad\text{for every }x\in\bit^n\text{ with }\idx{x}\leq k(n).
    \end{align}
    We claim that some $z\in\bit^{k(n)}$ has no good state. Suppose otherwise, and choose a good state $\rho_z$ for every $z$. Consider the following $k(n)\xmapsto{\epsilon(n)}m(n)$ quantum random access code: the encoding of $z$ is $\rho_z$, and the decoding measurement for an index $j\in[k(n)]$ consists of running $Q^{(i)}_n$ on the input $x_j\in\bit^n$ with $\idx{x_j}=j$, the advice register and fresh ancillas, and measuring the output qubit (this is a two-outcome measurement on the advice register). By goodness it recovers $z_j$ with probability at least $1-\epsilon(n)$ for every $z$ and $j$, so \thm{nayak} gives $m(n)\geq(1-H(\epsilon(n)))k(n)$, contradicting $n\in S$. Let $z^{(i)}\in\bit^{k(n_i)}$ be a string with no good state.

    Define $R$ by $R(x):=z^{(i)}_{\idx{x}}$ if $|x|=n_i$ for some $i$ and $\idx{x}\leq k(n_i)$, and $R(x):=0$ for all other $x$. Then $z_{n_i}(R)=z^{(i)}$ for every $i$.
    Suppose that $L:=L^{(k)}_R\in\BQTIME(t)_\epsilon/\q m$ for some $t$, via a quantum algorithm $Q$ and advice states $(\rho_n)_{n}$. Then $Q=Q^{(i)}$ for some $i$. At length $n:=n_i$, for every $x\in\bit^n$ with $\idx{x}\leq k(n)$ we have $L(x)=z^{(i)}_{\idx{x}}$ by \eq{L-vs-z} and $\Pr[Q^{(i)}_n(x,\rho_n)\neq L(x)]\leq\epsilon(n)$, so $\rho_n$ is good for $z^{(i)}$---a contradiction.
\end{proof}

\subsection{The upper bound}

For the matching upper bound we encode $z_n(R)$ with the random access code of \thm{dec-error-bound} and use the encoding as a randomized advice string.
The decoder needs to know $k(n)$ (to recognize the inputs $x$ with $\idx{x}>k(n)$) and the parameters $\ell$ and $d$ of \prot{toeplitz-rac}. Rather than assuming that these functions are computable within the time bound, we store them in a self-delimiting header of $O(\log k)$ bits. We also distinguish the RAC decoding budget $s(n)$ from the total running time $O(n^2+s(n))$. The resulting loss in the error bound is absorbed by \lem{stability}.

\begin{theorem}\label{thm:completeness-2}
    There are absolute constants $c_1,c\geq1$ with the following property.
    Let $m,k,s\colon\N\to\N$ with $k(n)\leq2^n$ for all $n$, and let $R\colon\bit^*\to\bit$ be arbitrary.
    Call $n$ \emph{admissible} if $m(n)\leq k(n)$ and $s(n)\geq c_1\log^2(2k(n))$; for admissible $n$ define
    \begin{align}
        \epsilon_1(n):=H^{-1}\!\left(1-\frac{m(n)}{k(n)}\right),\qquad
        \eta(n):=\frac{\log s(n)}{s(n)}+\frac{s(n)}{k(n)},
    \end{align}
    and let
    \begin{align}\label{eq:eps-prime}
        \epsilon'(n):=\begin{cases}
            \min\bigl\{1/2,\ \epsilon_1(n)+c\,\lossfn(\epsilon_1(n),\eta(n))\bigr\} & \text{if $n$ is admissible},\\[2pt]
            1/2 & \text{otherwise}.
        \end{cases}
    \end{align}
    Then $L^{(k)}_R\in\DTIME(n^2+s)_{\epsilon'}/\r m$.
\end{theorem}
\begin{proof}
    Let $c_0,c_2$ and $c'$ be the constants of \thm{dec-error-bound} ($c'$ is the constant called $c$ there), and set $c_1:=c_0$ and $c:=5(2c'+1)$.
    We describe, for every $n$, the advice distribution $D_n$ on $\bit^{m(n)}$ and the behaviour of a single machine $M$ on inputs of length $n$. The machine knows neither $m(n)$, $k(n)$ nor $s(n)$; the first advice bit $a_1$ will indicate whether a parameter header is present.
    Write
    \begin{align}
        w:=\floor{\log k(n)}+1,\qquad g:=4w.
    \end{align}
    For an integer $v\in\{0,\ldots,2^w-1\}$ let $\mathrm{bin}_w(v)\in\bit^w$ be its binary representation. The string $0^{w-1}\mathrm{bin}_w(k(n))$ is self-delimiting: its initial run of zeros determines $w$, since the representation of $k(n)$ starts with $1$.

    \emph{Non-admissible $n$, or $m(n)\leq g$.} If $m(n)=1$, let $D_n$ be uniform on $\bit$. If $m(n)\geq2$, let $D_n$ be the distribution of $a=(0,r,0^{m(n)-2})$ with $r$ a uniformly random bit.

    \emph{Admissible $n$ with $m(n)\geq g+1$.} Write $k=k(n)$, $m=m(n)$, $s=s(n)$, and let $\tilde m:=m-g\geq1$. Since $\tilde m\leq k$ and $s\geq c_0\log^2(2k)$, \thm{dec-error-bound}, applied to $(k,\tilde m,s)$, provides a $k\xmapsto{\tilde\epsilon}\tilde m$ random access code $(\Enc,\Dec)$ with
    \begin{align}\label{eq:tilde-eps}
        \tilde\epsilon\leq\min\bigl\{1/2,\ \tilde\epsilon_1+c'\lossfn(\tilde\epsilon_1,\eta)\bigr\},
        \qquad\tilde\epsilon_1:=H^{-1}\!\left(1-\frac{\tilde m}{k}\right),
    \end{align}
    which is either the trivial code or \prot{toeplitz-rac} with parameters $(\ell,d,B)$ as in \eq{parameters}, with $t=s$ there and $d<\ell<k$.
    The advice distribution $D_n$ is the distribution of
    \begin{align}
        a:=\bigl(1,\ 0^{w-1},\ \mathrm{bin}_w(k),\ \mathrm{bin}_w(\ell),\ \mathrm{bin}_w(d),\ \Enc(z_n(R),\omega)\bigr)\in\bit^{g+\tilde m}=\bit^m
    \end{align}
    for a random seed $\omega$, where in the case of the trivial code we set $\ell=d:=0$. The header has $1+(2w-1)+w+w=g$ bits. We call the last $\tilde m$ bits the \emph{payload}; it starts at position $g+1$.

    On input $x\in\bit^n$ and advice $a$, the machine $M$ proceeds as follows.
    \begin{enumerate}[label=(\roman*)]
        \item It reads $a_1$ and the content of the second advice cell. If the second cell is blank (the advice has length $1$), it outputs $a_1$. If $a_1=0$, it outputs $a_2$. These steps take $O(1)$ time.
        \item Otherwise ($a_1=1$), starting at the second advice cell it reads the initial run of zeros to determine $w$, and then reads the remaining bits of $k$ and the $w$-bit fields for $\ell$ and $d$. It computes $\idx{x}$ from the input; if $\idx{x}>k$, it outputs $0$. Parsing and these computations take $O(n^2+w)$ time.
        \item If $\ell=0$, it outputs the first payload bit (position $g+1$).
        \item Otherwise it runs $\Dec(\idx{x},\cdot)$ on the payload, with the parameters $\ell$ and $d$ read from the header. After discarding leading zeros from its binary representation, the index has at most $w$ bits. Also, $B\ell<2k$, $m'\leq\tilde m\leq k$ and $\ell,d<k$, so $w+1$ bounds the bit lengths needed in \lem{encoding-length}\ref{itm:el-b}. Adding the offset $g$ costs $O(w)$ per jump. Since the decoder
uses a constant number of jumps and reads each requested segment
sequentially, these accesses cost $O(d+w)$ time. Together with
the $O(w^2)$ arithmetic needed to compute the addresses, the
total running time is $O(d+w^2)$.
    \end{enumerate}

    \emph{Running time.} In the first case the machine outputs a uniformly random bit in $O(1)$ steps. In the case with a parameter header, $w\leq\log(2k)$ and $d<\ell\leq s/c_2$ for the nontrivial code, while $s\geq c_0\log^2(2k)$. Thus the machine halts within $O(n^2+s)$ steps on every advice string in the support of $D_n$.

    \emph{Error when no parameter header is present.} At non-admissible lengths the error is $1/2=\epsilon'(n)$. If $n$ is admissible but $m\leq g$, then
    \begin{align}
        1-H(\epsilon_1)=\frac{m}{k}\leq\frac{4w}{k}\leq\frac{4\log(2k)}{k}\leq\frac{4s}{k}\leq4\eta.
    \end{align}
    By \lem{entropy-perturbation} with $\epsilon_0=1/2$ and \ref{itm:P1},
    \begin{align}
        \frac{1}{2}-\epsilon_1\leq\lossfn(\epsilon_1,1-H(\epsilon_1))\leq4\lossfn(\epsilon_1,\eta)\leq c\lossfn(\epsilon_1,\eta).
    \end{align}
    Hence $\epsilon'(n)=1/2$ here as well. Thus insufficient space for the header requires no additional admissibility condition.

    \emph{Error when a parameter header is present.} If $\idx{x}>k$ then $M(x,a)=0=L^{(k)}_R(x)$ for every $a$ in the support. Otherwise, \eq{L-vs-z} and \eq{rac-condition} give $\Pr_{a\sim D_n}[M(x,a)\neq L^{(k)}_R(x)]\leq\tilde\epsilon$ (including the trivial-code case). Moreover,
    \begin{align}
        \epsilon_1\leq\tilde\epsilon_1,\qquad H(\tilde\epsilon_1)-H(\epsilon_1)=\frac{g}{k}\leq\frac{4s}{k}\leq4\eta.
    \end{align}
    Applying \lem{stability} with $\theta:=4\eta$ and $a:=c'$, followed by \ref{itm:P1}, yields
    \begin{align}
        \tilde\epsilon\leq\min\bigl\{1/2,\ \epsilon_1+(2c'+1)\lossfn(\epsilon_1,5\eta)\bigr\}\leq\min\bigl\{1/2,\ \epsilon_1+5(2c'+1)\lossfn(\epsilon_1,\eta)\bigr\}=\epsilon'(n).
    \end{align}
    Altogether, $M$ and $(D_n)_n$ witness $L^{(k)}_R\in\DTIME(n^2+s)_{\epsilon'}/\r m$.
\end{proof}

\subsection{The separations}

Combining the two theorems, allowing a slightly larger error while keeping the advice length fixed is strictly more powerful. We first state the general form, in which the RAC decoding budget is a parameter, and then specialize to polynomial time.

\begin{corollary}\label{cor:separation-2}
    Let $c_1,c$ be the constants of \thm{completeness-2}. Let $\epsilon\colon\N\to[0,1/2)$ and $m,s\colon\N\to\N$, and put
    \begin{align}\label{eq:k-choice}
        k(n):=\min\left\{\left\lfloor{\frac{m(n)}{1-H(\epsilon(n))}}\right\rfloor+1,\ 2^n\right\},\qquad
        \eta(n):=\frac{\log s(n)}{s(n)}+\frac{s(n)}{k(n)}.
    \end{align}
    Call $n$ \emph{good} if $m(n)<(1-H(\epsilon(n)))\,2^n$ and $s(n)\geq c_1\log^2(2k(n))$, and define
    \begin{align}\label{eq:delta-def}
        \delta(n):=\begin{cases}
            \min\bigl\{1/2,\ (2c+1)\,\lossfn\bigl(\epsilon(n),\ \eta(n)+1/m(n)\bigr)\bigr\} & \text{if $n$ is good},\\[2pt]
            1/2 & \text{otherwise}.
        \end{cases}
    \end{align}
    If infinitely many $n$ are good, then
    \begin{align}
        \DTIME(n^2+s)_{\epsilon+\delta}/\r m\not\subseteq\BQP_\epsilon/\q m .
    \end{align}
\end{corollary}
\begin{proof}
    Let $n$ be good. Then $m(n)/(1-H(\epsilon(n)))<2^n$, so $k(n)=\floor{m(n)/(1-H(\epsilon(n)))}+1\leq2^n$ and $k(n)>m(n)/(1-H(\epsilon(n)))\geq m(n)$; thus $m(n)<(1-H(\epsilon(n)))\,k(n)$, and $n$ is admissible in the sense of \thm{completeness-2}.
    Since infinitely many $n$ are good, \thm{soundness-2} provides an $R$ with $L:=L^{(k)}_R\notin\BQP_\epsilon/\q m$, while \thm{completeness-2} shows $L\in\DTIME(n^2+s)_{\epsilon'}/\r m$ with $\epsilon'$ as in \eq{eps-prime}.
    By \lem{basic}(a) it suffices to show $\epsilon'(n)\leq\epsilon(n)+\delta(n)$ for every $n$.
    If $n$ is not good, then $\delta(n)=1/2\geq\epsilon'(n)$.
    Let $n$ be good and drop the argument $n$. From $k\leq m/(1-H(\epsilon))+1$,
    \begin{align}
        \frac{m}{k}\geq\frac{m\,(1-H(\epsilon))}{m+1-H(\epsilon)}\geq\frac{m\,(1-H(\epsilon))}{m+1}=(1-H(\epsilon))\Bigl(1-\frac{1}{m+1}\Bigr)\geq1-H(\epsilon)-\frac{1}{m},
    \end{align}
    so $H(\epsilon_1)=1-m/k\leq H(\epsilon)+1/m$; and $H(\epsilon)<1-m/k=H(\epsilon_1)$ gives $\epsilon<\epsilon_1$. By \lem{stability} with $\theta:=1/m$ and $a:=c$,
    \begin{align}
        \epsilon'=\min\{1/2,\ \epsilon_1+c\lossfn(\epsilon_1,\eta)\}\leq\min\{1/2,\ \epsilon+(2c+1)\lossfn(\epsilon,\eta+1/m)\}\leq\epsilon+\delta ,
    \end{align}
    where the last step uses $\epsilon'\leq1/2\leq\epsilon+1/2$ together with $\epsilon'\leq\epsilon+(2c+1)\lossfn(\epsilon,\eta+1/m)$.
\end{proof}

\begin{corollary}\label{cor:bqp-separation-1}
    In the situation of \cor{separation-2}, if moreover $s(n)\leq n^{O(1)}$, then $\P_{\epsilon+\delta}/\r m\not\subseteq\BQP_\epsilon/\q m$.
\end{corollary}
\begin{proof}
    For polynomially bounded $s$, the total running time $n^2+s(n)$ is polynomial, so $\DTIME(n^2+s)_{\epsilon+\delta}/\r m\subseteq\P_{\epsilon+\delta}/\r m$.
\end{proof}

\begin{remark}\label{rmk:regimes}
    (a) The statements are meaningful only at lengths where $\epsilon(n)+\delta(n)<1/2$; if $\epsilon+\delta\geq1/2$ at some length, the classical class contains every language there (\lem{basic}(b)), which is consistent with, but says nothing beyond, the corollaries. For the choice $s\approx\sqrt{k\log k}$ discussed after \thm{dec-error-bound}, $\delta(n)\to0$ along good $n$ as $m(n)\to\infty$, as long as $\epsilon(n)$ stays bounded away from $1/2$. The $O(n^2)$ cost of processing the input does not enter the coding overhead $\eta$.
    (b) \cor{bqp-separation-1} allows $m(n)$ up to nearly $(1-H(\epsilon(n)))2^n$. For such long advice the polynomial-time machine can only use the advice through random access (\dfn{model}); when $m$ is polynomially bounded, random access is irrelevant up to polynomial factors.
    (c) \thm{soundness-2} rules out quantum algorithms of arbitrary running time, so the lower bound is information-theoretic; only the upper bound uses the efficiency of the decoder.
\end{remark}

Finally we specialize to polynomially bounded advice.

\begin{corollary}\label{cor:final}
    There is an absolute constant $c^*$ with the following property.
    Let $m\colon\N\to\N$ be polynomially bounded with $m(n)=\omega(1)$, and let $\epsilon\colon\N\to[0,1/2)$ satisfy $1/2-\epsilon(n)\geq n^{-b}$ for some constant $b$ and all sufficiently large $n$.
    Then there is a function $\delta\colon\N\to(0,1/2]$ with
    \begin{align}\label{eq:final-delta}
        \delta(n)\leq c^*\min\left\{\Bigl(\frac{\log m(n)}{m(n)}\Bigr)^{1/4},\ \frac{1}{1/2-\epsilon(n)}\sqrt{\frac{\log m(n)}{m(n)}}\,\right\}\qquad\text{for all sufficiently large }n,
    \end{align}
    such that
    \begin{align}
        \P_{\epsilon+\delta}/\r m\not\subseteq\BQP_\epsilon/\q m
        \qquad\text{and consequently}\qquad
        \BQP_\epsilon/\q m\subsetneq\BQP_{\epsilon+\delta}/\q m .
    \end{align}
    In particular, for every constant $\epsilon\in(0,1/2)$ the separation holds with $\delta(n)=O_\epsilon\bigl(\sqrt{\log m/m}\bigr)$ and $\epsilon+\delta(n)<1/2$ for all sufficiently large $n$. Uniformly in $\epsilon$, the gap satisfies $\delta(n)=O\bigl((\log m/m)^{1/4}\bigr)$.
\end{corollary}
\begin{proof}
    Define $k$ by \eq{k-choice}, let
    \begin{align}\label{eq:small-advice-budget}
        s(n):=\max\left\{\ceil{\sqrt{k(n)\log k(n)}},\ \ceil{c_1\log^2(2k(n))}\right\}
    \end{align}
    with $c_1$ from \thm{completeness-2}, and let $\delta$ be given by \eq{delta-def}; it is positive since $\eta(n)+1/m(n)>0$.

    \emph{All sufficiently large $n$ are good.} By \lem{entropy-perturbation} with $\epsilon_0=1/2$ and $\epsilon_1=\epsilon(n)$ we have $1/2-\epsilon(n)\leq\kappa\sqrt{1-H(\epsilon(n))}$, hence $1-H(\epsilon(n))\geq\kappa^{-2}n^{-2b}$ for large $n$. As $m$ is polynomially bounded, $m(n)<\kappa^{-2}n^{-2b}2^n\leq(1-H(\epsilon(n)))2^n$ for large $n$. The condition $s(n)\geq c_1\log^2(2k(n))$ holds by construction, so all sufficiently large $n$ are good in \cor{separation-2}.
    Consequently $k(n)=\floor{m(n)/(1-H(\epsilon(n)))}+1\leq\kappa^2n^{2b}m(n)+1$ for large $n$, so $k$, $s$ and the total running time $n^2+s(n)$ are polynomially bounded. Thus \cor{bqp-separation-1} yields $\P_{\epsilon+\delta}/\r m\not\subseteq\BQP_\epsilon/\q m$.
    By \lem{basic}(c), $\P_{\epsilon+\delta}/\r m\subseteq\BQP_{\epsilon+\delta}/\q m$, hence $\BQP_{\epsilon+\delta}/\q m\not\subseteq\BQP_\epsilon/\q m$; together with $\BQP_\epsilon/\q m\subseteq\BQP_{\epsilon+\delta}/\q m$ (\lem{basic}(a)) this gives the proper inclusion.

    \emph{The bound on $\delta$.} Since $k(n)\geq m(n)\to\infty$ for all sufficiently large $n$, eventually $\sqrt{k\log k}\geq c_1\log^2(2k)$ and hence $s=\ceil{\sqrt{k\log k}}$. In particular,
    $\sqrt{k\log k}\leq s\leq3\sqrt{k\log k}$ and $\log s\leq2\log k$ for all sufficiently large $n$. Therefore
    \begin{align}
        \frac{\log s}{s}\leq\frac{2\log k}{\sqrt{k\log k}}=2\sqrt{\frac{\log k}{k}},
        \qquad
        \frac{s}{k}\leq\frac{3\sqrt{k\log k}}{k}=3\sqrt{\frac{\log k}{k}},
    \end{align}
    so $\eta\leq5\sqrt{(\log k)/k}\leq5\sqrt{(\log m)/m}$, because $x\mapsto(\log x)/x$ is decreasing for $x\geq3$ and $k\geq m$. Moreover $1/m\leq\sqrt{(\log m)/m}$ for $m\geq2$. Hence $\eta+1/m\leq\xi:=6\sqrt{(\log m)/m}$, and by \ref{itm:P1} and \ref{itm:P4},
    \begin{align}
        \delta(n)\leq(2c+1)\,\lossfn\Bigl(\epsilon,\eta+\frac{1}{m}\Bigr)
        &\leq(2c+1)\min\left\{\kappa\sqrt\xi,\ \frac{\kappa^2\xi}{1/2-\epsilon}\right\}\nonumber\\
        &\leq c^*\min\left\{\Bigl(\frac{\log m}{m}\Bigr)^{1/4},\ \frac{1}{1/2-\epsilon}\sqrt{\frac{\log m}{m}}\,\right\}
    \end{align}
    with $c^*:=(2c+1)\max\{\kappa\sqrt6,\ 6\kappa^2\}$. For fixed $\epsilon<1/2$, the second bound tends to zero because $m(n)\to\infty$, giving the final assertion.
\end{proof}

\section*{AI statement}
AI-assisted tools, primarily ChatGPT 5, and ChatGPT 6, and Fable 5.1, were used extensively in the development of this work. Beyond assistance with writing and presentation, AI contributed to the technical development of the paper. In particular, AI assistance was used to adapt the diagonalization argument, and to propose the Toeplitz-matrix-based coding construction used to reduce the advice overhead. AI was also used to explore alternative constructions and parameter choices, derive and check intermediate bounds, identify potential gaps in proofs, and improve the presentation of the resulting arguments. 


These AI-generated ideas and derivations were subsequently examined, modified where necessary, and independently checked by the authors before being incorporated into the manuscript. The authors take responsibility for the correctness of all statements and proofs in the final manuscript.

\section*{Acknowledgments}

SHH is supported by National Science and Technology Council in Taiwan, Grant~No.~114-2222-E-002-002-MY3, and by National Taiwan University, Grant~No.~NTU-CC-114L895503.
HHL is supported by NSTC under project numbers 115-2112-M-007-006 and 114-2221-E-007-083-MY3, and by funding of Taiwan Centers of Excellence (TCE).

\bibliographystyle{plainurl}
\bibliography{LLM-pieces/separation_advice_references}

\end{document}